\documentclass[12pt]{amsart}

\usepackage{amssymb, verbatim, amsthm, amsfonts, amsmath}

\newtheorem{lemma}{Lemma}

\newtheorem{theorem}{Theorem}

\newtheorem{corollary}{Corollary}
\newtheorem{remark}{Remark}

\begin{document}

\title{A generalization of the map $\chi$}



\author{Xiutao Feng, Qiang Wang, Jingyi Yu, Anpeng Zhang }
\thanks{The authors were supported by the National Key Research and Development Program of China (Grant No. 2026ZD150070), CAS Project for Young Scientists in Basic Research (Grant No. YSBR-035),  and  the Natural Sciences and Engineering Research Council of Canada
(RGPIN-2023-04673) }

\address{State Key Laboratory of Mathematical Sciences, Academy of Mathematics and Systems Science, Chinese Academy of Sciences, Beijing, 100190, China}
\email{fengxt@amss.ac.cn}

\address{School of Mathematics and Statistics, Carleton University, 1125
Colonel By Drive, Ottawa, ON K1S 5B6, Canada}
\email{wang@math.carleton.ca}

\address{College of Science, Civil Aviation University of China, Tianjin, 300300, China}
\email{409575860@qq.com}

\address{The High School Affiliated to Renmin University of China, Beijing, 100872, China}
\email{zhanganpeng@amss.ac.cn}

\maketitle

\begin{abstract}
The map  $ \chi_n:\mathbb{F}_2^n \to \mathbb{F}_2^n$ defined by $y=\chi_n(x)$ with $y_i = x_i + x_{i+1}x_{i+2} + x_{i+2}$, where the indices are computed modulo $n$, has been widely studied for its application in lightweight cryptography.  In this paper, we generalize this map  and completely characterize all these shift-invariant permutations of the form $y_i=x_{i+u}+x_{i+v}(x_{i+w}+a_i)$, where $0\le u, v, w<n$ and $a_i\in \mathbb{F}_2$, $1\le i\le n$.  We also determine their compositional inverses explicitly. 
\end{abstract}

\section{Introduction}

In order to find a shift-invariant transformation over $\mathbb{F}_2^n$ that is easy to implement and has good cryptographic properties,  Daemen \cite{Daemen1995} introduced a map
$ \chi_n:\mathbb{F}_2^n \to \mathbb{F}_2^n$ defined by $y=\chi_n(x)$ with 
\[
y_i = x_i + x_{i+1}x_{i+2} + x_{i+2},  \qquad 1 \le i \le n,
\]
 where the indices are computed modulo $n$.  
 It is known that $\chi_n:\mathbb{F}_2^n \to \mathbb{F}_2^n$ is bijective if and only if $n$ is odd \cite{Daemen1995}. The map $\chi_n$ has been used in several cryptographic primitives, including Keccak-f \cite{Nis08} in the NIST SHA-3 standard \cite{Nis15} and the winner of the NIST lightweight competition ASCON \cite{Nis25}. 
A closed-form expression for $\chi_n^{-1}$ was obtained by Liu et al. \cite{MR4493714}; for odd $n>0$, it is given by
\[
x_i = y_i + \sum_{j=1}^{(n-1)/2} y_{i-2j+1} \prod_{k=j}^{(n-1)/2} (y_{i-2k}+1),
\]
with indices computed modulo $n$.

 The order and cycle structure of $\chi_n$ are also known; see  Schoone and Daemen \cite{MR4741123}. They further studied algebraic properties of this map via a univariate polynomial representation under the isomorphism $\mathbb{F}_{2^n}\cong \mathbb{F}_2^n$. They also considered $\chi_n$ as a polynomial map and showed that the variant defined by
$
y_i = x_i + (x_{i+1}+1)x_{i+2}
$
does not define a bijection over extensions of $\mathbb{F}_p$ for odd $p$, nor over characteristic two extensions whose degree is divisible by $2$ or $3$. Based on these observations, they conjectured that this rule does not yield a bijection over any extension field of $\mathbb{F}_2$. This conjecture was recently confirmed by Graner et al. \cite{MR5014643} using a linear-algebraic approach. 

Various variations and generalizations of $\chi_n$ have been studied in \cite{ALPS2025, MR4909542,    KK25,MR4826048,  MR5047775, Ras2022}. For example, for odd $n$, Kriepke and Kyureghyan \cite{MR4826048} showed that all iterates of $\chi_n$ can be expressed as linear combinations of the maps $\gamma_{2k}$ for $k\ge1$, where the $i$-th coordinate of $\gamma_{2k}$ is
\[
y_i = x_{i+2k}(x_{i+2k-1}+1)(x_{i+2k-3}+1)\cdots(x_{i+1}+1), \ i\in\{0,\ldots,n-1\}.
\]
They studied the linear span $\Gamma_n$ of the $\gamma_{2k}$ over $\mathbb{F}_2$ and showed that every permutation in $\Gamma_n$ lies in
\[
G_n = \gamma_0 + \mathrm{span}\{\gamma_2,\ldots,\gamma_{n-1}\}.
\]
Moreover, $G_n$ forms a monoid under composition, isomorphic to
\[
M_n = \left\{\,1 + \sum_{i=1}^t a_i X^i \,\right\} \subseteq \mathbb{F}_2[X]/(X^{(n+1)/2}),
\]
that is, to the unit group of the residue ring $\mathbb{F}_2[X]/(X^{(n+1)/2})$. This isomorphism enables both theoretical and algorithmic constructions of a rich family of shift-invariant permutations and reduces the study of their properties to corresponding binary polynomials.

For even $n$, Kriepke and Kyureghyan \cite{KK25} showed that the subset of bijections in $\Gamma_n$ forms an abelian group isomorphic to the unit group of $\mathbb{F}_2[X]/(X^n + X^{n/2})$. Their approach was further adapted by Lyu et al. \cite{MR5047775} to construct $\chi$-like permutations. They introduced a generalized class of maps $\chi_{n,m}$ defined by $y=\chi_{n,m}(x)$, where
\[
y_i = x_i + x_{i+m}(x_{i+m-1}+1)(x_{i+m-2}+1)\cdots(x_{i+1}+1).
\]
For $m=2$, $\chi_{n,2}$ coincides with $\chi_n$. They further showed that $\chi_{n,m}$ lies in an abelian group isomorphic to the unit group of $\mathbb{F}_2[X]/(X^{(\lfloor n/m \rfloor+1)/2})$ when $m$ does not divide $n$.

In this paper, we  introduce another generalization of $\chi_n$. Let $n=2^k n_0$, where $k\geq 0$ and $n_0$ is odd.  For any $ 1 \leq  v  < n$, define $\chi_{n, v}:\mathbb{F}_2^n \to \mathbb{F}_2^n$ by $y=\chi_{n, v} (x)$ with
\[
y_i = x_i + x_{i+v}x_{ i+2v} + x_{i+2v}, \qquad 1 \le i \le n.
\]
Then we will show that $\chi_{n, v}$ permutes $\mathbb{F}_2^n$ if and only if $2^k \mid v$.  Under the same condition, we introduce another class of permutations of  $\mathbb{F}_2^n$,   $\chi_{n, -2v}:\mathbb{F}_2^n \to \mathbb{F}_2^n$, defined by 

\[
y_i = x_i + x_{i-2v}x_{ i-v} + x_{i-v}, \qquad 1 \le i \le n. 
\]

The original  map $\chi_n$ is a special case of  $\chi_{n, v}$ when $v=1$.  It permutes  $\mathbb{F}_{2}^n$ if and only if $n$ is odd.  These two newly introduced  classes of  nonlinear maps induce permutations of $\mathbb{F}_{2}^n$ where $n$ can be even.  Moreover, their  latency are  lower than  that of the recent generalization introduced by Lyu et al.  \cite{MR5047775}  because the tree depth of circuit implementations on both $\chi_{n, v}$ and $\chi_{n, -2v}$ is exactly two, the same as that of the original $\chi_n$. 

In order to derive these two classes of nonlinear permutations,  we study all these shift-invariant maps of algebraic degree $2$ of the same shape.  Namely,  let  
$\chi_{n, u, v, w}: \mathbb{F}_2^n \to \mathbb{F}_2^n$ defined by $y=\chi_{n,  u, v, w} (x)$ with
\begin{equation}\label{map}
y_i = x_{i+u} +  x_{i+v} (x_{i+w} + a_i),  \qquad 1 \le i \le n.
\end{equation}
In Section 2, we completely characterize when  $P(x) = \chi_{n, u, v, w}(x) $ permutes $\mathbb{F}_{2}^n$. Namely,  $\chi_{n, u, v, w}$ permutes $\mathbb{F}_{2}^n$ if and only if  $2^k \mid  (v-w)$, 
$u-w \equiv 2 (v-w) \pmod n $ or $u-w \equiv - (v-w) \pmod n$, and $a_0=\cdots = a_{n-1} =1$.  Essentially all such permutations are shift equivalent to either   $\chi_{n, v}$ or $\chi_{n, -2v}$.  In Section 3, we determine the inverse of $P(x) = \chi_{n, u, v, w}(x) $ explicitly.

\section{Low-latency Permutations}

Let  $\ggg$ and $\lll$ be the left and right shift operator respectively.  We rewrite the map  (\ref{map}) as 
	\[
		y = (x\ggg u) + (x\ggg v) \cdot ((x\ggg w)+ a),
	\]
	where $x, y, a\in \mathbb{F}_2^n$, and $u, v, w$ are integers. Without  loss of generality, we can assume that $1\leq u\ne v < n$ and $w=0$ since $y=P(x)$ is a permutation if and only if $y=P(x\lll w)$ is a permutation, that is,  $y=P(x)$ is defined by
	\begin{equation} \label{eq:pe}
		y = P(x) := (x\ggg u) + (x\ggg v) \cdot (x + a).
	\end{equation}
	Here it should be noticed that if $u=v$, then $y=(x\ggg v) \cdot (x+a+\mathbf{1})$ follows. Therefore $P(x)$ is a permutation if and only if $v=0$ and $a=\mathbf{1}$, that is, $P(x)=x$ is an identity. Below we always assume that $u\ne v$.
	
  Let $0\ne \delta\in \mathbb{F}_2^n$. Then $y=P(x)$ is a permutation if and only if for any nonzero $\delta$, the equation on $x$
	\[
		((x+ \delta)\ggg u) + ((x+ \delta)\ggg v) \cdot (x+ \delta+ a) = (x\ggg u) + (x\ggg v) \cdot (x+ a)
	\]
	has no solution in $\mathbb{F}_2^n$.  The equation can be simplified to 
	\begin{equation}\label{eq:pcond}
		(\delta \ggg v)\cdot x+ \delta\cdot(x\ggg v)=(\delta\ggg v)\cdot \delta+(\delta\ggg v)\cdot a+(\delta\ggg u).
	\end{equation}
	We rewrite it with its bit form using the following system of linear equations: 
	\begin{equation}\label{eq:pbcond}
		\delta_{i+v} x_i+\delta_i x_{i+v} = \delta_{i+v}\delta_i+\delta_{i+v} a_i+\delta_{i+u}, \ i=0,1,\cdots,n-1,
	\end{equation}
	where all operations on the subscripts are done modulo $n$.
	
	Let $d=\text{gcd}(v, n),  \tau=\text{gcd}(d, u)$ and $n=d\ell$. We define
	\[
	C_i = \left\{(i+t v)\pmod n \left|   t=0,1,\cdots, \ell-1\right.\right\}, i=0,1,\cdots, d-1
	\]
	and
	\[
	D_j = \left\{(j+tu)\pmod d \left|   t=0,1,\cdots, \frac{d}{\tau}-1\right.\right\}, j=0,1,\cdots, \tau-1.
	\]
	Denote
	\[
	T_j = \bigcup_{i\in D_j} C_i
	\]
	for $0\le j<\tau$.	Then we  can divide Eqn. (\ref{eq:pbcond}) into $\tau$ independent subsystems $E_0,\cdots, E_{\tau-1}$,  and for each subsystem $E_j$, the subscripts of its all parameters $x, \delta, a$ are in $T_j$:
	\begin{equation} \label{eqn5}
		E_j:\quad 		\delta_{i+v} x_i+\delta_i x_{i+v} = \delta_{i+v}\delta_i+\delta_{i+v} a_i+\delta_{i+u}, \ i\in T_j.
	\end{equation}
	
	 Obviously, that Eqn. (\ref{eq:pbcond}) has no solution for all nonzero $\delta$ is equivalent to that every subsystem $E_j$ defined in (\ref{eqn5})  has no solution for all possible nonzero $\delta_{T_j}$, where $\delta_{T_j}$ means all bits of $\delta$ whose coordinates are in $T_j$.  Indeed,  if a subsystem $E_j$  has a solution for some nonzero $\delta_{T_j}$, then we can construct a nonzero $\delta'$ such that Eqn. (\ref{eq:pbcond}) also has a solution, where $\delta'_i=\delta_i$ for $i\in T_j$ and $\delta'_i=0$ for $i\not\in T_j$. Therefore we focus on  one subsystem  $E_j$ for some $j$. Without loss of generality, below we will only consider $E_0$.
	
	In order to  exclude some values of $a$, we usually choose a certain $\delta$  such that $\delta_{T_0}$ are not zero and $E_0$ has solutions.  	 For simplification, for a given nonempty ordered set $A=\{i_0, i_1, \cdots, i_{|A|-1}\} \subseteq \mathbb{Z}_n$, we denote $\delta_A=(\delta_{i_0}, \delta_{i_1},\cdots, \delta_{i_{|A|-1}})\in \mathbb{F}_2^{|A|}$.	
	
	For any fixed index $i \in D_0$,  we denote $C_i=(c_{i,0}, c_{i,1}, \cdots, c_{i,l-1})$, where $c_{i,t}=i+tv\pmod  n$, $t=0,1, \cdots, \ell-1$. Further, we simply write $x_{c_{i,t}}, \delta_{c_{i,t}}, a_{c_{i,t}}$ as $x_{i,t}, \delta_{i,t}, a_{i,t}$  respectively.  We can further focus on  the subsystem 
	\begin{equation} \label{eq:ci}
		\delta_{i,t+1} x_{i,t}+\delta_{i,t} x_{i,t+1} = \delta_{i,t}\delta_{i,t+1}+\delta_{i,t+1} a_{i,t}+\delta_{j, s+t},  
	\end{equation}
	where $t=0,1\cdots,\ell-1$, $j=(i+u)\pmod  d$ and $s= (\frac{i+u-j}{d})(\frac{v}{d})^{-1}\pmod \ell$. 
	
Below we consider the following cases  so that  Eqn. (\ref{eq:ci})  has a solution.

	\begin{enumerate}
		\item[(a)]  If $\ell=2$, then it is easy to check that Eqn. (\ref{eq:ci}) always has two solutions for any nonzero $\delta_{C_i}$.
		\item[(b)]  If  there is no zero entry in $\delta_{C_i}$, i.e., $\delta_{C_i}=(11\cdots1)$, then Eqn. (\ref{eq:ci}) is written simply as
		\[
			x_{i,t}+x_{i,t+1}=1+a_{i,t}+\delta_{j,s+t}, t=0,1\cdots,\ell-1,
		\]
		and it has at least one solutions if and only if
		\begin{equation}\label{eq:cond2}
			\sum_{t=0}^{\ell-1} a_{i,t} \equiv \sum_{t=0}^{\ell-1} \delta_{j,t} +\ell \pmod 2.
		\end{equation}
		In particular,  if  $u\in C_0$, that is, $j=i  \pmod  d$,  then Eqn. (\ref{eq:cond2}) is equivalent to
		\begin{equation}\label{eq:cond2_}
			\sum_{t=0}^{\ell-1}a_{i,t} \equiv 0 \pmod 2.
		\end{equation}
		\item[(c)]  Assume $\delta_{C_i}=(1\cdots101\cdots1)$, i.e.,  $\delta_{i,k}=0$ and $\delta_{i,t}=1$ for $t\ne k$. Then  we have
		\[
		\left\{\begin{array}{l}
			x_{i,k} = \delta_{j,s+k-1} = a_{i,k}+\delta_{j,s+k},\   t=k-1, k, \\
			x_{i,t}+x_{i,t+1}=1+a_{i,t}+\delta_{j,s+t},\  t\ne k-1, k.
		\end{array}	\right.
		\]
		It is easy to see that  Eqn. (\ref{eq:ci})  has solutions if and only if
		\begin{equation} \label{eq:cond3}
			a_{i,k} = \delta_{j,s+k-1}+\delta_{j,s+k}.
		\end{equation}
		\item[(d)]  Assume $\delta_{C_i}=(1\cdots10\cdots01\cdots1)$, i.e., $\delta_{i,k}=\delta_{i,k+1}=\cdots=\delta_{i,k+r-1}=0$ and $\delta_{i,k}=1$ for all the other $k$'s, where $r\ge2$. Similarly to Item (c), 	Eqn. (\ref{eq:ci}) has solutions if and only if 
		\begin{equation} \label{eq:cond4}
			\delta_{j,s+k}=\cdots=\delta_{j,s+k+r-2}=0.
		\end{equation}
	\end{enumerate}

Before deriving the main theorem, we first present several  useful lemmas.

\begin{lemma} \label{lem: setI1}
Let $n, v,u$ be defined as above. If $\gcd(v,n)\nmid u$, then there exists a nonempty set $I\subseteq\mathbb Z_n$ such that
\begin{equation}\label{eq:pcond_setI}
	I\cap(I+v)=\varnothing,
	\qquad
	I+u\subseteq I\cup(I+v).
\end{equation}	
\end{lemma}

\begin{proof}
We consider directed walks on $\mathbb Z_n$ with the following two allowed step types:
\[
\mathsf A=u,\qquad \mathsf B=u-v.
\]
Since repeated steps of type $\mathsf A$ eventually return to the starting vertex,  such closed walks exist.   Choose a nonempty closed walk of minimum length with both types of steps are allowed,  By minimality, its vertices are distinct except for the coinciding initial
and terminal vertices.   Among all such walks, choose one walk $\mathcal{W}$ having the maximum number of changes between the two step types,  where the last and first steps are also regarded as consecutive.
 Let $I$ be its vertex set.

For every $i\in I$, its successor is either $i+u$ or $i+u-v$. Thus
\[
I+u\subseteq I\cup(I+v).
\]

Suppose there exists $j\in I$ that both $j$ and $j+v$ belong to $I$. The segment from $j$ to $j+v$  in the walk $\mathcal{W}$  has displacement $v$. If this segment contains a step of type $\mathsf A$, replacing that step by $\mathsf B$ would decrease its displacement by $v$.  Indeed, there is a closed walk from $j$ to $j$ within this segment.  This would produce a shorter nonempty closed walk containing $j$, contradicting minimality of $\mathcal{W}$. Hence this segment consists entirely of steps of type $\mathsf B$.

Similarly, the complementary segment from $j+v$ to $j$ in  the walk $\mathcal{W}$   consists entirely of steps of type $\mathsf A$. Therefore, without loss of generality, we can assume that the cyclic step sequence  $\mathcal{W}$  has the form
\[
\mathsf B^{t_1}\mathsf A^{t_2}
\]
for some positive integers $t_1,t_2$. Here $ t_1 (u-v) \equiv v \pmod{n}$ and  $t_2 u \equiv -v \pmod n$.  Because   $\gcd(v,n)\nmid u$, we obtain $t_1, t_2\ge2$.  Hence we can rearrange the steps of the walk $\mathcal{W}$  into the order
\[
\mathsf B\,\mathsf A\,
\mathsf B^{t_1-1}\mathsf A^{t_2-1}.
\]
The resulting closed walk  $\mathcal{W^\prime}$ still contains $j$ and has the same minimum length. However, the number of changes of step type in the walk $\mathcal{W}^\prime$  is four,
whereas that of $\mathcal{W}$  is two. This contradicts our choice of the walk $\mathcal{W}$. Thus $I\cap(I+v)=\varnothing$.
\end{proof}

\begin{lemma} \label{lem:usigma}
	For any given $n\ge5$ and $3\le s\le n-2$, let $\sigma_0=0$ and $\sigma_{i+1}=\sigma_i+s-1 \pmod  n$ or $\sigma_i+s \pmod  n$ for $i\ge0$. Then there exists a cycle
	\[
	0=\sigma_0\rightarrow \sigma_1 \rightarrow \sigma_2 \rightarrow \cdots \rightarrow \sigma_{r-1} \rightarrow \sigma_{r}=0
	\] 
	such that $\sigma_i\ne \sigma_j\pm1 \pmod n$ for any $0\le i< j< r$.
\end{lemma}
\begin{proof}
	Since a walk tarting from 0 with step $s$ eventually returns to 0,  such a cycle always exists. We assume that for any given cycle 
	\[
	0=\sigma_0\rightarrow \sigma_1 \rightarrow \cdots \rightarrow  \sigma_{i} \rightarrow\cdots \rightarrow \sigma_{j-1} \rightarrow \sigma_j \rightarrow \sigma_{j+1} \rightarrow\cdots\rightarrow \sigma_{r}=0
	\]
	it always exist two positions $\sigma_i$ and $\sigma_j$ such that $\sigma_i\equiv\sigma_j\pm 1\pmod n$. Let $\mathcal{W}$ be a shortest one.
	For simplicity, we make convention that all operations on the subscripts are done modulo $r$, for example, $\sigma_{r+1}=\sigma_{r+1\pmod r}=\sigma_1$, and operations on integers are done modulo $n$, for example,  $\sigma_{i+1}=\sigma_i+s$ means $\sigma_{i+1}=\sigma_i+s\pmod n$.
	
	Without loss of generality, we assume that $\sigma_j=1$ for some $j>1$.  Indeed,  if  $\sigma_j=\sigma_i+1 $ for some $i>0$,  then we take $\sigma'_{k}=\sigma_{k}-\sigma_i$ for $0\le k<r$ and thus construct a new shortest cycle $\mathcal{W}'$  such that $\sigma'_i=0$ and $\sigma'_{j} =1$.   
	
	Now we  prove that $\sigma_{k+1}=\sigma_k+s$ for any $j\le k<r$.  It is because: if there exists a $k$ such that $\sigma_{k+1}=\sigma_k+s-1$, we take the first $k$ satisfying the above condition.  Then we can construct a shorter cycle $\mathcal{W}'$:
	\[
		0=\sigma_0\xrightarrow{s}  \sigma'_{j+1} \xrightarrow{s}\cdots\xrightarrow{s}\sigma'_{k}\xrightarrow{s}\sigma_{k+1}\rightarrow\cdots\rightarrow \sigma_{r}=0,
	\]
	which contradict with the assumption.
	
	Secondly,  we prove that $\sigma_{k+1}=\sigma_k+s-1$ for any $0\le k<j$. If there exists a $k$ such that $\sigma_{k+1}=\sigma_k+s$, then we take this step $s-1$ instead of $s$, and arrive  $\sigma_j$ with displacement 0, which is a shorter cycle than $\mathcal{W}$.
	
	Combining the above discussions, we concluded that  all shortest invalid cycles must have the form
	\begin{equation}\label{eq:invalidf}
		0\xrightarrow{s-1} \sigma_1  \xrightarrow{s-1} \cdots  \xrightarrow{s-1} \sigma_{j-1} \xrightarrow{s-1} 1 \xrightarrow{s} \sigma_{j+1} \xrightarrow{s} \cdots \xrightarrow{s} 0.		
	\end{equation}
	However, if  we make a minor adjustment on the above cycle, then we obtain a new cycle 
	\[
	0\xrightarrow{s-1} \sigma_1  \xrightarrow{s-1} \cdots  \xrightarrow{s-1} \sigma_{j-1} \xrightarrow{s} 2 \xrightarrow{s-1} \sigma_{j+1} \xrightarrow{s} \cdots \xrightarrow{s} 0.
	\]
	This new cycle is also a shortest cycle but it does not have the form (\ref{eq:invalidf}).  A contradiction happens.  Therefore the conclusion holds. 
\end{proof}

\begin{remark}
	In Lemma \ref{lem:usigma}, the conditions $n\ge5$ and $3\le s\le n-2$ are necessary. This is because: When $s=1$, it is obvious that there does not exist such a nonadjacent walk; if $s=2$, then the step $s-1$ becomes 1. Thus one has exactly one choice  if he wants a nonadjacent walk, that is, step $s$. When $n$ is odd, such a walk dost not exist. Similarly, if $s=n-1$, it is equivalent that one walk step 1 or 2 backward. Thus it requires $3\le s\le n-2$ if there exists a nonadjacent walk on $\mathbb{Z}_n$. Thus $n-2\ge s\ge3$, that is, $n\ge5$.
\end{remark}
	
       Now we  derive some necessary conditions on $u, v,  a$ so that  $y=P(x)$ is a permutation of $\mathbb{F}_2^n$, that is, Eqn. (\ref{eq:pbcond}) have no solution for any nonzero $\delta$.
	
	\begin{lemma}  \label{lemma1}
		If  $P(x) = (x\ggg u) + (x\ggg v) \cdot (x + a)$   is a permutation of $\mathbb{F}_2^n$,  
		then $u\in C_0$.  
	\end{lemma}
	\begin{proof}
		We assume that $u\not\in C_0$, that is,  $\gcd(v,n)\nmid u$.  By Lemma \ref{lem: setI1}, there exists a nonempty set $I\subset \mathbb Z_n$ such that $I\cap (I+v) = \varnothing$ and $(I+u) \subset I\cup (I+v)$. 	Consider the Boolean component
		\[
		g(x)=\sum_{i\in I}P_i(x)
		=\sum_{i\in I}
		\bigl(x_ix_{i+v}+a_ix_{i+v}+x_{i+u}\bigr).
		\]
		By $I\cap(I+v)=\varnothing$, the variable pairs
		\[
		\{x_i,x_{i+v}\},\qquad i\in I,
		\]
		are pairwise disjoint. Moreover, $I+u\subseteq I\cup(I+v)$ ensures that every term $x_{i+u}$ in $g$ can be written as one of $x_i$ or $x_{i+v}$. Hence there exist $\alpha_i,\beta_i\in\mathbb{F}_2$ such that
		\[
		g(x)=\sum_{i\in I}
		\bigl(x_ix_{i+v}+\alpha_ix_i+\beta_ix_{i+v}\bigr).
		\]
		
      Since
		\[
		\sum_{s,t\in\mathbb{F}_2}
		(-1)^{st+\alpha s+\beta t}
		=2(-1)^{\alpha\beta},
		\]
		the disjointness of the variable pairs yields
		\[
		\begin{aligned}
			\sum_{x\in\mathbb{F}_2^n}(-1)^{g(x)}
			&=
			2^{n-2|I|}
			\prod_{i\in I}
			\left(
			\sum_{s,t\in\mathbb{F}_2}
			(-1)^{st+\alpha_i s+\beta_i t}
			\right)\\
			&=
			2^{n-|I|}(-1)^{\sum_{i\in I}\alpha_i\beta_i}
			\ne0.
		\end{aligned}
		\]
		
		On the other hand, since $P$ is a permutation and $I\ne\varnothing$, changing variables $y=P(x)$ gives
		\[
		\sum_{x\in\mathbb{F}_2^n}(-1)^{g(x)}
		= \sum_{y\in\mathbb{F}_2^n} (-1)^{\sum_{i\in I}y_i}
		=0.
		\]
		This is a contradiction. Therefore $u\in C_0$.
	\end{proof}

	\begin{lemma} \label{lemma2}
		Let $u\in C_0$.  If  $P(x) = (x\ggg u) + (x\ggg v) \cdot (x + a)$   is a permutation of $\mathbb{F}_2^n$,  
		 then $a=\mathbf{1}$, the all one vector. 
	\end{lemma}
	
	\begin{proof}
Recall that  if  $P(x) = (x\ggg u) + (x\ggg v) \cdot (x + a)$   is a permutation of $\mathbb{F}_2^n$ then
 Eqn. (\ref{eq:pcond})  or  Eqn. (\ref{eq:pbcond}) has no solution on $x$ in $\mathbb{F}_2^n$ for any nonzero $\delta$.  It is noticed that $d\mid u$ and $\tau=d$ due to $u\in C_0$, that is,  each $C_i$ forms a subsystem $E_i$.
		Assume that $a\ne\mathbf{1}$. Firstly, if $a=\mathbf{0}$, then we take $\delta=\mathbf{1}$ and thus  $\mathbf{0}$ and $\mathbf{1}$ are two solutions to Eqn. (\ref{eq:pcond}), a contradiction. So we have $a\ne\mathbf{0}$. In this case, we  consider the following two subcases  according to the value of  $d=\text{gcd}(n, v)$. 
		\begin{itemize}
			\item $d=1$.  We take $i$ such that $a_{i}=0$ due to $a\ne\mathbf{1}$. Further we take $\delta_i=0$ and $\delta_t=1$ for  all $t\ne i$. Since
			\[
				0=a_{i}=\delta_{i-v+u}+\delta_{i+u}=0,
			\]
			thus Eqn. (\ref{eq:ci}) has solutions according to Eqn. (\ref{eq:cond3}).   So Eqn. (\ref{eq:pbcond}) has solutions. It leads to a contradiction. 
			
			\item $d>1$.  Note that $a\ne \mathbf{1}$, if there exists an $i$ such that $a_{C_i}=\mathbf{0}$, then we take $\delta_{C_i}= \mathbf{1}$ and conclude  that Eqn. (\ref{eq:ci})  has solutions by Eqn. (\ref{eq:cond2_}). It leads to a contradiction. Thus there must exist one group $C_i$ such that $a_{C_i}$ is neither  the zero vector nor the all one vector.  In this case, similarly to the case $d=1$,  we take $j\in C_i$ such that $a_j=0$. Further we take $\delta_j=0$ and $\delta_t=1$ for all $j\ne t\in C_i$.  Since
			\[
				0=a_{j}=\delta_{j-v+u}+\delta_{j+u}=0,
			\]
			thus Eqn. (\ref{eq:ci}) has solutions according to Eqn. (\ref{eq:cond3}), which shows that the subsystem $E_i$ has solutions and leads to a contradiction.
			
		\end{itemize}
		Combining the above two cases, we complete the proof. 
	\end{proof}
		
		\begin{lemma} \label{lemma3}
		Let $n=2^k n_0$, $k\ge0$, $n_0$ odd, $u\in C_0$ and $a=\mathbf{1}$.  If  $P(x) = (x\ggg u) + (x\ggg v) \cdot (x + a)$   is a permutation of $\mathbb{F}_2^n$,  
		then $2^k\mid v$.
	\end{lemma}
	\begin{proof}
		If $2^k\nmid v$, then $2^k \nmid d$ and thus $ \ell=\frac{n}{d}$ is even. Take $\delta=\mathbf{1}$. Then Eqn. (\ref{eq:cond2_}) holds. Hence 
	 Eqn. (\ref{eq:pbcond}) has at least one solution, a contradiction.  Hence the conclusion holds.
	\end{proof}

\medskip	
	
In summary,  if $P(x) =  (x\ggg u) + (x\ggg v) \cdot (x + a)$ permutes  $\mathbb{F}_{2}^n$ such that $n=2^k n_0$ and $n_0$ is odd, then it is necessary to assume that  $u = tv ~\text{mod}~n$ for some $t$,  $a=\mathbf{1}$, and $2^k \mid v$.  Now we are ready to prove our main result. 	
	
	\begin{theorem} \label{th:p1}
		Let $n=2^k n_0$, $k\ge 0$, $n_0$ odd, and $1\le u, v<n$. Then
		\begin{equation} \label{eq:thp1}
		P(x) = (x\ggg u)+(x\ggg v)\cdot (x+a)
		\end{equation}
		is a permutation over $\mathbb{F}_2^n$ if and only if the following three conditions holds:
		\begin{enumerate}
			\item $2^k \mid v$;
			\item $u= 2v\pmod n$ or $u=(n-1)v \pmod n$;
			\item  $a=\mathbf{1}$.  
		\end{enumerate}
	\end{theorem}
	
	\begin{proof}
		We first prove the necessity.  Assume that $P(x)$ is a permutation.  By Lemmas~\ref{lemma1},  \ref{lemma2}, \ref{lemma3},  we  have $u=s v\pmod n$, $2^k\mid v$ and $a=\mathbf{1}$. Now we further prove $u=n-v$ or $u=2v\pmod n$. We assume that $u\ne n-v$ and $u \ne 2v\pmod n$. Since Eqn. (\ref{eq:pbcond}) is divided into $d$ independent subsystems $E_i$ that all indexes of  its parameters are in $C_i$, we need to construct a nonzero $\delta$ such that each $E_i$ has solutions. 	Note that $|C_i| = \ell = \frac{n}{d}$ is an odd for any $i$ and we only need to prove that any one of them has solutions. For this reason, we can assume that without loss of generality  $n$ is odd and $d=(v, n)=1$. 
		
		Let $u=s v\pmod n$. For any $3\le s\le n-2$, by Lemma \ref{lem:usigma}, there exists a cycle
				\[
			0=\sigma_0\rightarrow\sigma_1\rightarrow\cdots\rightarrow\sigma_r=0
		\]
		such that $\sigma_i\ne \sigma_j\pm 1\pmod n$ for any $0\le i\ne j<r$. We take
		\[
			\delta_{kv} = \left\{\begin{array}{ll}
				0 & \text{if } k\in\{\sigma_0,\sigma_1,\cdots,\sigma_{r-1}\}, \\
				1 &  \text{otherwise}.
			\end{array}\right.
		\]
		and obtain
		\[
		1=\delta_{\sigma_0 v}+\delta_{(\sigma_0+1)v}=\delta_{\sigma_1v}+\delta_{(\sigma_1+1)v}=\cdots=\delta_{\sigma_{r-1}v}+\delta_{(\sigma_{r-1}+1)v}
		\]
		due to $a=\mathbf{1}$ and Eqn. (\ref{eq:cond3}). In this case, $\delta$ has a zero entry such that the entries before and after are both nonzero. 
		Similar to the arguments in proving Eqn. (\ref{eq:cond3}), Eqn. (\ref{eq:pbcond}) has solutions, contradicts to that $P(x)$ is a permutation.  Thus $s$ must be equal to $2$ or $n-1$.
		
		Conversely,  we prove that $P(x)$ is a permutation  under the above assumptions. First we  prove the case $s=2$. We proceed with the proof by contradiction.   Assume that Eqn. (\ref{eq:pbcond}) has solutions, whose coefficients are $\delta_i \delta_{i+v}\delta_{i+2v}$.   Since $|C_0|= \ell = n$ is odd, we have $\delta\ne\mathbf{1}$ by Eqn. (\ref{eq:cond2}).  For simplicity, we rearrange all bits of $\delta$ as $\delta_0 \delta_{v}\delta_{2v}\cdots \delta_{(n-1)v}$ and still write it as $\delta$. Suppose that there are two adjacent zeros in $\delta$. Without loss of generality, let $\delta_0=\delta_1=0$. Then by Eqn. (\ref{eq:pbcond}), we have $\delta_2=0$, and then $\delta_k=0$ for all $k\ge3$, a contradiction. Moreover, if $\delta_i=0$ for some $i$, then $\delta_{i-1}=\delta_{i+1}=1$ and $x_{i}=1=\delta_{i+1}+\delta_{i+2}$ by Eqn. (\ref{eq:pbcond}), that is, $\delta_{i+2}=0$.  Hence $\delta$ must have the form $10101\cdots01$. But it leads to a contradiction $1=a_{n-1}=\delta_{n-1}+\delta_0=0$, which implies that Eqn. (\ref{eq:pbcond}) has no solution for any nonzero $\delta$. Finally, when $s=n-1$, i.e., $u=-v\pmod n$, we have the coefficients of Eqn. (\ref{eq:pbcond}) are $\delta_i\delta_{i+v}\delta_{i-v}$. Similarly, we earrange all bits of $\delta$ as $\delta_{-v}\delta_0 \delta_{v}\delta_{2v}\cdots \delta_{(n-2)v}$ and still write it as $\delta$.  Then $\delta$ must have the form $10101\cdots01$, too. It also leads to a contradiction. So the conclusion holds.
	\end{proof}
	
	Here it should be pointed that the map corresponding to $u\equiv-v\pmod n$ is linearly equivalent, by cyclic coordinate permutations and reversal, to the one corresponding to $u\equiv2v\pmod n$. Indeed, we define the linear permutation
\[
R:\mathbb F_2^n\longrightarrow\mathbb F_2^n,
\qquad
(R(x))_i=x_{-i},
\]
where all indices are taken modulo $n$. Clearly $R^{-1}=R$.

It is known that the reversal map changes a right cyclic shift into the corresponding
left cyclic shift. More precisely, for every integer $r$,
\[
R(x\ggg r)=R(x)\ggg(-r).
\]
Moreover, since $R$ merely permutes the coordinates,
\[
R(x\cdot y)=R(x)\cdot R(y)
\]
for the component-wise product, and $R(\mathbf{1})=\mathbf{1}$.
Let $z=R(x)$. Applying $R$ to $P(x)$ gives
\begin{align*}
	R(P(x))
	&=R(x\ggg(-v))+R\bigl((x\ggg v)(x+\mathbf{1})\bigr)\\
	&=(z\ggg v)+(z\ggg(-v))(z+\mathbf{1}).
\end{align*}
Therefore
\[
R\circ P\circ R^{-1}(z)=(z\ggg v)+(z\ggg(-v))(z+\mathbf{1}).
\]

Set
\[
v'=-v\pmod n.
\]
Then
\[
2v'\equiv -2v\pmod n.
\]
The first shift in the preceding map is $v=-v'$, rather than $2v'$,
so one more cyclic output shift is convenient. Let
\[
S(y)=y\ggg(-3v).
\]
Since $S$ is a coordinate permutation, it does not affect
bijectivity. We obtain
\begin{align*}
	S\bigl(R(P(R^{-1}(z)))\bigr)
	&=
	(z\ggg(-2v))
	+(z\ggg(-4v))
	\bigl((z\ggg(-3v))+\mathbf{1}\bigr).
\end{align*}
After the change of variables
\[
X=z\ggg(-3v),
\]
this becomes a map of the form
\[
(X\ggg 2v)+(X\ggg v)(X+\mathbf{1}),
\]
up to a relabeling of the cyclic coordinates.
Since all these coordinate transformations are bijections, the former is a permutation whenever the latter is.

\begin{corollary} \label{cor1}
	Let $n=2^k n_0$, $n_0$ odd and $\mathcal{N}_n$ be the number of permutations $P(x)$ of the form  (\ref{eq:thp1}). Then $\mathcal{N}_n=2n_0-4$ if $3|n_0$ and $2n_0-2$ if $3\nmid n_0$.
\end{corollary}
\begin{proof}
	By Theorem \ref{th:p1}, we let $v=2^k v_0$ and $u=2^k u_0$. Then $u_0=2v_0\pmod n_0$ or $n_0-v_0$. The conclusion follows from that $v_0$ can be taken from 1 to $n_0-1$ and there are two corresponding $u_0$ for each $v_0$,  except that if $3 \mid n_0$ and $v_0 = \frac{n_0}{3}$  then   $2v_0 = n- v_0 = \frac{2n_0}{3}$. 
\end{proof}

\begin{corollary}
	If $n$ is a power of 2, then $P(x)$ is not a permutation for any $u, v$ and $a$.
\end{corollary}

Finally, it is straightforward to obtain the following result. 

	\begin{theorem} \label{thm2}
		Let $n=2^k n_0$, $k\ge 0$, $n_0$ odd, and $1\le u, v, w<n$. Then
		\begin{equation} \label{eq:generalp}
		P(x) = (x\ggg u)+(x\ggg v)\cdot (x\ggg w +a)
		\end{equation}
		is a permutation over $\mathbb{F}_2^n$ if and only if the following three conditions holds:
		\begin{enumerate}
			\item $2^k \mid (v-w)$;
			\item $u-w= 2(v-w)\pmod n$ or $u-w=-(v-w) \pmod n$;
			\item  $a=\mathbf{1}$.  
		\end{enumerate}
	\end{theorem}

\begin{remark}
We note that  $ \chi_{n,v}$ can be viewed a stretching  of  the original  map $\chi$ in the sense of  Proposition 1.3 in  \cite{HauglandOmland26}, which   considers
$f_s(x_1, ..., x_{(k-1)s+1}) = f(x_1, x_{s+1}, ..., x_{(k-1)s+1})$ for a local rule $f$ of diameter $k$ and an integer $s$ at least $2$.   The proof splits $\mathbb{F}_2^n$ into $\gcd(n,v)$ interleaved tracks of length $n/\gcd(n,v)$, and the induced map acts as $\chi$ on each track separately. Together with Daemen's criterion this gives that $\chi_{n,v}$ permutes $\mathbb{F}_2^n$ if and only if $n/\gcd(n,v)$ is odd, which is the condition $2^k \mid v$. So our  permutations in even dimension are interleaved copies of $\chi$ on odd tracks.  This distinguishes them from the siblings of Kriepke and Kyureghyan \cite{KK25}, which do not decompose in this way.
\end{remark}

\begin{remark}
The notion of  {\it elementary equivalence} was introduced in  Definition 2.10 in \cite{HauglandOmland25}.   We note that  $\chi_{n,v}$ and $\chi_{n,-2v}$ are elementary equivalent in their sense. 
One is obtained from the other by reversing the order of the variables, complementing them, and adding $1$. Elementary equivalent rules induce maps with identical cryptographic parameters. 
\end{remark}

\begin{remark}
The behaviour of all iterates of $\chi_{n, -2v}$  and $\chi_{n, v}$ is similar. 
\end{remark}

\begin{remark}
Theorem ~\ref{th:p1} and Corollary \ref{cor1} completely characterize 
these shift-invariant permutations of algebraic degree $2$ of the  shape
$y_i = x_{i+u} + x_{i+v}(x_{i+w} + a_i)$,
with a single quadratic monomial.
There are quadratic shift-invariant permutations outside our shape. The appendix of \cite{HauglandOmland25} lists four rules of algebraic degree $2$ and diameter $6$ that are
 permutations of $\mathbb{F}_2^n$ for $n$ in $\{11, 13, 15, 17, 19\}$, for instance  $x_1 + x_2 + x_3 + x_5 + x_2x_3 + x_4x_5 + x_5x_6$.
\end{remark}

\section{The inverse of  $\chi_{n, u, v, w}$}

In this section we shall give the explicit inverse of the generalized map $\chi_{n,  u, v, w}$.  Keep the notations from the previous section.  First we consider the inverse of the stepped map $\chi^{(t)}$, which is defined as
\begin{equation}\label{eq:chi-def}
	y_i=\chi^{(t)}(x)_i=x_i+(x_{i+t}+1)x_{i+2t},	\qquad i\in\mathbb Z_n,
\end{equation}
where $t\in\mathbb{Z}_n$.

\begin{lemma} \label{lem:stepped-chi}
	Let $n\ge3$, $t\in\mathbb Z_n$ and $\chi^{(t)}$ be defined as above. Denote $d=\gcd(n,t)$ and $\ell=n/d$.
	If $\ell$ is odd ,	then $\chi^{(t)}$ is a permutation and its compositional inverse is
	\begin{equation}\label{eq:chi-inverse}
	x_i=y_i+\sum_{j=1}^{m}y_{i-(2j-1)t}\prod_{k=j}^{m}(y_{i-2kt}+1),
	\end{equation}
	where $m=(\ell-1)/2$ and subscripts are computed modulo $n$. When
	$\ell=1$, the sum is empty and the map $\chi^{(t)}$  is the identity.
\end{lemma}

\begin{proof}
	The subgroup generated by $t$ in $\mathbb Z_n$ has order
	$\ell$. Its cosets
	\[
	C_r=\{r,r+t,\ldots,r+(\ell-1)t\},\qquad 0\le r<d,
	\]
	partition $\mathbb Z_n$. Fix one such coset and relabel its coordinates by
	\[
	X_j=x_{r+jt},\qquad Y_j=y_{r+jt},\qquad j\in\mathbb Z_\ell.
	\]
	The defining equations become
	\[
	Y_j=X_j+(X_{j+1}+1)X_{j+2},\qquad j\in\mathbb Z_\ell.
	\]
	This is exactly the classical map $\chi$ on an odd cycle.
	For odd $\ell$, the classical inverse formula of
	Liu--Sarkar--Meier--Isobe gives
	\[
	X_j=Y_j+\sum_{h=1}^{m}Y_{j-2h+1}
	\prod_{k=h}^{m}(Y_{j-2k}+1)
	\]
    Replacing $X_j,Y_j$ by their original coordinates yields
	\[
	x_{r+jt}=y_{r+jt}+\sum_{h=1}^{m}
	y_{r+(j-2h+1)t}\prod_{k=h}^{m}
	(y_{r+(j-2k)t}+1),
	\]
	which is the asserted formula at $i=r+jt$.
	Since the cosets are disjoint, the inverse works independently on each
	coset and hence on the whole space. For $\ell=1$, $t\equiv0\pmod n$
	and $x_i+(x_i+1)x_i=x_i$, as claimed.
\end{proof}

\begin{theorem} \label{thm:inverse}
	Let $P$ be given by Eqn. \eqref{eq:thp1} under the permutation conditions above.
	\begin{enumerate}
		\item If $u\equiv-v\pmod n$, so that
		\[
		y_i=x_{i-v}+x_{i+v}(x_i+1),
		\]
		then $x=P^{-1}(y)$ is given by
		\begin{equation}\label{eq:inverse-minus}
			{x_i=y_{i+v}+\sum_{j=1}^{m}y_{i+(2-2j)v}
				\prod_{k=j}^{m}\bigl(y_{i+(1-2k)v}+1\bigr).}
		\end{equation}
		\item If $u\equiv2v\pmod n$, so that
		\[
		y_i=x_{i+2v}+x_{i+v}(x_i+1),
		\]
		then $x=P^{-1}(y)$ is given by
		\begin{equation}\label{eq:inverse-plus}
			{x_i=y_{i-2v}+\sum_{j=1}^{m}
				\bigl(y_{i+(2j-3)v}+1\bigr)
				\prod_{k=j}^{m}y_{i+(2k-2)v}.}
		\end{equation}
	\end{enumerate}
\end{theorem}

\begin{proof}
	By Lemma~\ref{lem:stepped-chi}, the inverse formula
	\eqref{eq:chi-inverse} applies with step $v$ or $-v$ on each
	odd-length cycle of the shift by $v$.  Indeed, $2^k\mid v$
	implies that $\ell=n/\gcd(n,v)$ is odd.  We now reduce each
	case to the stepped $\chi$-map.
	
	\medskip
	\noindent\textbf{Case 1: $u\equiv-v\pmod n$.}
	The defining equations are
	\[
	y_i=x_{i-v}+x_{i+v}(x_i+1).
	\]
	To put the variables in the standard $\chi$ order, replace $i$
	by $i+v$.  This yields
	\begin{equation}\label{eq:minus-shift}
		y_{i+v}=x_i+(x_{i+v}+1)x_{i+2v}.
	\end{equation}
	Define $Y_i=y_{i+v}$.  Eqn. \eqref{eq:minus-shift} becomes
	$Y_i=\chi^{(v)}(x)_i$.  Thus Eqn. \eqref{eq:chi-inverse} with
 $t=v$, gives
	\[
	x_i=Y_i+\sum_{j=1}^{m}Y_{i-(2j-1)v}
	\prod_{k=j}^{m}(Y_{i-2kv}+1).
	\]
	Now calculate each index after substituting $Y_i=y_{i+v}$:
	\[
	Y_{i-(2j-1)v}=y_{i-(2j-2)v}
	=y_{i+(2-2j)v},
	\qquad
	Y_{i-2kv}=y_{i+(1-2k)v}.
	\]
   Hence
	\[
	x_i=y_{i+v}+\sum_{j=1}^{m}y_{i+(2-2j)v}
	\prod_{k=j}^{m}(y_{i+(1-2k)v}+1),
	\]
	which is \eqref{eq:inverse-minus}.  Conversely, the classical
	inverse identity shows that these $x_i$ satisfy
	\eqref{eq:minus-shift} and therefore the original equations.
	Uniqueness follows because $\chi^{(v)}$ and the output shift
	$y\mapsto Y$ are bijections.
	
	\medskip
	\noindent\textbf{Case 2: $u\equiv2v\pmod n$.}
	Here
	\begin{equation}\label{eq:plus-start}
		y_i=x_{i+2v}+x_{i+v}(x_i+1).
	\end{equation}
	First replace $i$ by $i-2v$:
	\[
	y_{i-2v}=x_i+x_{i-v}(x_{i-2v}+1).
	\]
	This expression has the reversed cyclic order.  Complement
	both the input and the shifted output by setting
	\[
	z_i=x_i+1,\qquad Y_i=y_{i-2v}+1.
	\]
	Since $x_i=z_i+1$, $x_{i-v}=z_{i-v}+1$, and
	$x_{i-2v}+1=z_{i-2v}$, we obtain
	\[
	\begin{aligned}
		Y_i
		&=x_i+1+x_{i-v}(x_{i-2v}+1)\\
		&=z_i+(z_{i-v}+1)z_{i-2v}\\
		&=\chi^{(-v)}(z)_i.
	\end{aligned}
	\]
	Apply Eqn. \eqref{eq:chi-inverse} with $t=-v$:
	\[
	z_i=Y_i+\sum_{j=1}^{m}Y_{i+(2j-1)v}
	\prod_{k=j}^{m}(Y_{i+2kv}+1).
	\]
   Consequently,
	\[
	Y_{i+(2j-1)v}=y_{i+(2j-3)v}+1,\qquad
	Y_{i+2kv}+1=y_{i+(2k-2)v}.
	\]
	Finally $x_i=z_i+1$, and the two constant terms in
	$Y_i+1=y_{i-2v}$ cancel in characteristic two.  Therefore
	\[
	x_i=y_{i-2v}
	+\sum_{j=1}^{m}(y_{i+(2j-3)v}+1)
	\prod_{k=j}^{m}y_{i+(2k-2)v},
	\]
	which is just Eqn. \eqref{eq:inverse-plus}.  The transformation
	$x\mapsto z$ is a coordinate-wise complement, and
	$y\mapsto Y$ is a cyclic shift followed by a complement;
	both are bijective.  Since $\chi^{(-v)}$ is bijective on each
	odd cycle, the displayed formula is the unique compositional
	inverse of Eqn. \eqref{eq:plus-start}.
	
	The two cases cover the permutation parameters of the theorem,
	so the proof is complete.
\end{proof}

\begin{corollary} \label{cor:degree}
	For either permutation in Theorem~\ref{thm:inverse}, the inverse has algebraic degree
	\[
	{\deg(P^{-1})=m+1=\frac{\ell+1}{2}
		=\frac{n/\gcd(n,v)+1}{2}.}
	\]
\end{corollary}
\begin{proof}
	Each summand in \eqref{eq:inverse-minus} and \eqref{eq:inverse-plus}
	has degree at most $m+1$. In either expression, the summand with $j=1$
	has a degree-$(m+1)$ monomial formed by taking every variable from its
	$m+1$ distinct factors. The other summands have at most $m$ factors,
	so this monomial cannot cancel. The $m+1$ indices are distinct modulo
	$n$ because they lie on a cycle of length $2m+1$.
\end{proof}

\paragraph{Example ($n=6$, $v=2$).}
Here $d=2$, $\ell=3$, and $m=1$. Since $2v\equiv-v\pmod 6$,
the two cases coincide. The map is
\[
y_i=x_{i+4}+x_{i+2}(x_i+1),
\]
and its inverse can be written in either equivalent form
\[
{x_i=y_{i-4}+(y_{i-2}+1)y_i=y_{i+2}+y_i(y_{i-2}+1).}
\]

Finally we return to the generalized map $\chi_{n, u, v, w}$. According to Theorem \ref{thm2}, under its hypotheses the map $\chi_{n, u, v, w}$ is a permutation precisely when $a_i=1$ for all $i$,
$2^k\mid(v-w)$, and
\[
u-w\equiv 2(v-w)\quad\text{or}\quad u-w\equiv-(v-w)\pmod n,
\]
where $n=2^k n_0$ and $n_0$ is odd. Set
\[
t=v-w\pmod n,\qquad d=\gcd(n,t),\qquad
\ell=\frac nd,\qquad m=\frac{\ell-1}{2}.
\]
We assume $t\not\equiv0\pmod n$, so that $\ell\ge3$ is odd.

\begin{theorem} \label{thm:general-inverse}
	Let $P$ be given by Eqn. \eqref{eq:generalp} and satisfy the permutation
	conditions above. Then its compositional inverse is as follows.
	\begin{enumerate}
		\item If $u-w\equiv-t\pmod n$, then
		\begin{equation}\label{eq:general-minus}
			{(P^{-1}(y))_i=y_{i-w+t}
				+\sum_{j=1}^{m}y_{i-w+(2-2j)t}
				\prod_{k=j}^{m}\bigl(y_{i-w+(1-2k)t}+1\bigr).}
		\end{equation}
		\item If $u-w\equiv2t\pmod n$, then
		\begin{equation}\label{eq:general-plus}
			{(P^{-1}(y))_i=y_{i-w-2t}
				+\sum_{j=1}^{m}\bigl(y_{i-w+(2j-3)t}+1\bigr)
				\prod_{k=j}^{m}y_{i-w+(2k-2)t}.}
		\end{equation}
	\end{enumerate}
	All subscripts are interpreted modulo $n$.
\end{theorem}
\begin{proof}
	Let $R_w:\mathbb F_2^n\to\mathbb F_2^n$ be the cyclic shift
	$(R_wx)_i=x_{i+w}$.  It is a linear permutation with
	$R_w^{-1}=R_{-w}$.  Define $z=R_wx$, so that
	\[
	z_i=x_{i+w},\quad
	x_{i+u}=z_{i+u-w},\quad
	x_{i+v}=z_{i+v-w}.
	\]
	The permutation hypothesis forces $a_i=1$ for every $i$.
	Thus the original three-shift map can be rewritten as
	\begin{equation}\label{eq:general-reduction}
		y_i=z_{i+u-w}+z_{i+t}(z_i+1),
		\qquad t=v-w.
	\end{equation}
	Write $Q(z)$ for the right-hand side of
	\eqref{eq:general-reduction}.  Then $P=Q\circ R_w$,
	and composition in reverse order gives
	\[
	P^{-1}=R_{-w}\circ Q^{-1}.
	\]
	In coordinates this is
	\begin{equation}\label{eq:general-index}
		(P^{-1}(y))_i=(Q^{-1}(y))_{i-w}.
	\end{equation}
	
	The hypotheses $2^k\mid t$ and
	$u-w\equiv-t$ or $2t\pmod n$ show that $Q$ satisfies
	the conditions of Theorem~\ref{thm:inverse}, with $v$
	replaced by $t$.  Its coordinate cycles have length
	$\ell=n/\gcd(n,t)$, which is odd.
	
	If $u-w\equiv-t\pmod n$, Theorem~\ref{thm:inverse} gives
	\[
	(Q^{-1}(y))_h=y_{h+t}
	+\sum_{j=1}^{m}y_{h+(2-2j)t}
	\prod_{k=j}^{m}(y_{h+(1-2k)t}+1).
	\]
	Setting $h=i-w$ as required by
	\eqref{eq:general-index} proves \eqref{eq:general-minus}.
	
	If $u-w\equiv2t\pmod n$, the other formula gives
	\[
	(Q^{-1}(y))_h=y_{h-2t}
	+\sum_{j=1}^{m}(y_{h+(2j-3)t}+1)
	\prod_{k=j}^{m}y_{h+(2k-2)t}.
	\]
	Again substituting $h=i-w$ proves \eqref{eq:general-plus}.
	The inverse is unique because both $Q$ and $R_w$ are
	permutations.  This completes the proof.
\end{proof}

\begin{corollary} \label{cor:general-degree}
	Under the hypotheses of Theorem~\ref{thm:general-inverse},
	\[
	 {\deg(P^{-1})=\frac{\ell+1}{2}
		=\frac12\left(\frac{n}{\gcd(n,v-w)}+1\right).}
	\]
\end{corollary}
\begin{proof}
	The inverse of $P$ differs from the inverse of $Q$ only by a
	permutation of output coordinates. Such a permutation preserves
	algebraic degree, so the claim follows from Corollary~\ref{cor:degree}.
\end{proof}

\section{Acknowledgement}
We want to thank Tron Omland  for very helpful suggestions and comments on the connection with the references \cite{HauglandOmland25, HauglandOmland26}.  We also thank Svetla Nikova for her suggestion which motivated us to add the results on the inverses into the original draft.

\end{document}